\documentclass[10pt,aps,twocolumn,longbibliography,footnoteemail]{revtex4-1}

\usepackage{amsmath}
\usepackage{amsthm}
\usepackage{amsfonts}
\usepackage{amssymb}
\usepackage{mathrsfs}

\usepackage{graphicx}
\usepackage{xcolor}

\usepackage{hyperref} 
\hypersetup{pdfborder={0 0 0},colorlinks=true}

\theoremstyle{definition}

\theoremstyle{theorem}
\newtheorem{lemma}{Lemma}
\newtheorem*{lemma*}{Lemma}
\newtheorem{proposition}{Proposition}
\newtheorem*{proposition*}{Proposition}

\newtheorem*{conjecture*}{Conjecture}

\newtheorem*{corollary*}{Corollary}
\theoremstyle{remark}

\newcommand{\U}{\mathscr{U}}

\renewcommand{\H}{\mathscr{H}}
\renewcommand{\P}{\mathscr{P}}
\renewcommand{\S}{\mathscr{S}}

\newcommand{\M}{\mathscr{M}}
\newcommand{\II}{\mathbb{I}}

\newcommand{\Hv}{\operatorname{H}}
\newcommand{\Id}{\operatorname{I}}

\newcommand{\n}{\pmb{n}}

\renewcommand{\a}{\pmb{a}}
\renewcommand{\b}{\pmb{b}}

\renewcommand{\sup}{\operatorname{sup}}

\newcommand{\abs}[1]{\left| #1\right|}

\newcommand{\norm}[1]{\left\| #1 \right\|}

\renewcommand{\d}{\operatorname{d}}
\newcommand{\Tr}{\operatorname{Tr}}

\newcommand{\cell}{\operatorname{box}}

\begin{document}
\title{Necessary and sufficient condition for steerability of two-qubit states by the geometry of steering outcomes}
\author{H. Chau Nguyen}
\email{chau@pks.mpg.de}
\affiliation{Max-Planck-Institut f\"ur Physik komplexer Systeme, \\ N\"othnitzer Stra{\ss}e 38, D-01187 Dresden, Germany}
\author{Thanh Vu}
\email{tvu@unl.edu}
\affiliation{Department of Mathematics, University of Nebraska-Lincoln, \\Lincoln, NE 68588, USA}
\begin{abstract}
Fully characterizing the steerability of a quantum state of a bipartite system has remained an open problem ever since the concept of steerability was first defined. In this paper, using our recent geometrical approach to steerability, we suggest a necessary and sufficient condition for a two-qubit state to be steerable with respect to projective measurements. To this end, we define the critical radius of local models and show that a state of two qubits is steerable with respect to projective measurements from Alice's side if and only if her critical radius of local models is less than $1$.  As an example, we calculate the critical radius of local models for the so-called T-states by proving the optimality of a recently-suggested ansatz for Alice's local hidden state model.
\end{abstract}

\maketitle

\section{Introduction}
Quantum steering, or Einstein-Podolsky-Rosen steering~\cite{Einstein1935a} is a fascinating phenomenon of quantum physics, yet has remained a topic of debate throughout the history of quantum mechanics~\cite{Bell1993a}. It indicates the ability of Alice to steer Bob's system into different ensembles from a distance. 
Though this description of quantum steering was known since the beginning of the debate~\cite{Einstein1935a,Schrodinger1935a}, it was not appreciated that quantum steering implies a precise classification of bipartite quantum states into steerable and unsteerable classes until recently~\cite{Wiseman2007a,Jones2007b}. Quantum steerability has since formed a third distinct class of quantum non-locality, in addition to two other well-known classes, non-separability~\cite{Werner1989a} and Bell non-locality~\cite{Bell1964a}.

Ever since steerability was defined, much effort has been spent on characterizing the steerable class of quantum states. Steering inequalities were developed as sufficient conditions for steerability~\cite{Cavalcanti2009a, Saunders2011a,Zukowski2015a,Marciniak2015a,Kogia2015a,Zhu2015a}. Recently, strategies to build local hidden state models to establish necessary conditions for steerability have also been suggested~\cite{Jevtic2015a,Bowles2016a,Cavalcanti2015a,Hirsch2015a}. The equivalences between steerability with joint-measurability in the theory of quantum measurements~\cite{Quintino2014a,Uola2014a}, and with sub-channel discrimination in quantum communication theory~\cite{Piani2015a} were also established. However, to our best knowledge, an intrinsic necessary and sufficient condition for steerability is still lacking even for the simplest case of two-qubit states and by projective measurements. The goal of this paper is to propose such a criterion in this simplest case. 

This work is a continuation of our geometrical framework established recently~\cite{Nguyen2016a}, where we simplify the problem of determining the steerability of a two-qubit state to classifying the double-cone of steering outcomes in the 4-dimensional (4D) Euclidean space. To this end, we will define Alice's critical radius of local models and show that a two-qubit state is steerable from Alice's side by projective measurements if and only if this radius is smaller than $1$; see Proposition \ref{prop:steering-radius}. Using our criterion, we prove a necessary and sufficient condition conjectured by Jevtic~\textit{et al.}~\cite{Jevtic2015a} on the steerability of the so-called T-states.

\section{Characterizing steerability geometrically}
\subsection*{EPR map and steering outcomes}
In this subsection, we recall the notion of EPR map and steering outcomes introduced in \cite{Nguyen2016a}. Consider a state of two qubits AB, that is, a positive (semi-definite) unit-trace operator $\rho$ over $\H_A \otimes \H_B$, where $\H_A$ and $\H_B$ are $2$-dimensional (2D) Hilbert spaces. The Euclidean spaces of hermitian operators acting on $\H_A$ and $\H_B$ are denoted by $B^H (\H_A)$ and $B^{H} (\H_B)$, respectively. Operators acting on $\H_A$ will be denoted by $A$, or labelled by super/sub-scripts A; and an analogous convention is applied to B. In particular, $\II^A$ and $\II^B$ are the identity operators over $\H_A$ and $\H_B$, respectively. The EPR map from Alice side $\rho^{A \to B}:B^H(\H_A) \to B^H(\H_B)$ is defined by $\rho^{A \to B} (A)=\Tr_A \left[ \rho (A \otimes \II^B)\right]$. If $X$ is a subset, or an element, of $B^H(\H_A)$, then its image under Alice's EPR map is denoted by $X'$. In particular, the set of Alice's measurement outcomes, $\M_A= \{M \in B^{H} (\H_A) \vert O \le M \le \II^A\}$ with $O$ denoting the zero operator, is mapped to the set of Alice's steering outcomes $\M_A' \subseteq B^H(\H_B)$. Moreover, let $\P_A$ be Alice's Bloch hyperplane, $\P_A= \{A  \in B^{H} (\H_A) \vert \Tr (A)=1\}$, and $\S_A$ be Alice's Bloch ball, $\S_A= \P_A \cap \M_A$. Their images under Alice's EPR map are $\P_A'$ and $\S_A'$ respectively.   

Note that $B^H(\H_A)$ and $B^H(\H_B)$ are 4D Euclidean spaces. If one chooses an orthonormal basis for $\H_A$, and uses the Pauli matrices $\{\sigma_i^{A}\}_{i=0}^{3}$ (with $\sigma_0^A= \II^A$) as the basis of $B^H(\H_A)$, any operator $A$ of $B^H(\H_A)$ can be written as $A=\frac{1}{2}\sum_{i=0}^{3} X_i(A) \sigma_i$, where $X_i(A) = \Tr (A \sigma_i^A)$. Similar orthogonal coordinates can be introduced for $B^H(\H_B)$. With these two coordinate systems, a density operator $\rho$ can then be written as $\rho= \frac{1}{4}\sum_{i,j=0}^{3} \Theta_{ij} (\rho) \sigma^A_i \otimes \sigma^B_j$, where $\Theta_{ij} (\rho)=\Tr [\rho (\sigma_i^A \otimes \sigma_j^B) ]$. Accordingly, Alice's EPR map is presented by the $(4 \times 4)$ matrix $\frac{1}{2} \Theta^T$.

\subsection*{Packability of steering outcomes and steerability}
A projective measurement performed on A is a decomposition of $\II^A$ into $P^A_1+P^A_2$, where $P^A_1$ and $P^A_2$ are two orthogonal projections. A state $\rho$ is called \emph{unsteerable} from Alice's side with respect to projective measurements if there exists a decomposition of $(\II^A)'=\Tr_A(\rho)$ into an ensemble of $m$ positive operators of $\H_B$, $(\II^A)'= \sum_{i=1}^{m} B_i$, such that for any projective measurement $\{P^A_1,P^A_2\}$ on A, the corresponding conditional ensemble of Bob's system B, $\{(P^A_i)'=\Tr_A [\rho (P^A_i \otimes \II^B)]\}_{i=1}^2$, can be expressed as a stochastic map from $\{B_i\}_{i=1}^m$ to $\{(P^A_1)',(P^A_2)'\}$, i.e.,
\begin{equation}
(P^A_i)' = \sum_{j=1}^{m} G_{ij} B_j,
\label{eq:unsteerable-def}
\end{equation}
where $G$ is a $(2 \times m)$ stochastic matrix, $0 \le G_{ij} \le 1$, $\sum_{i=1}^{2} G_{ij} = 1$. The ensemble $\{B_i\}_{i=1}^m$ is called  Alice's ensemble of \emph{local hidden states}. Alice's local hidden states, if they exist, will allow her to locally simulate quantum steering on Bob's system, and therefore quantum steering cannot be verified; see~\cite{Wiseman2007a,Jones2007b}. In this paper, we only consider steerability with respect to projective measurements, which will be called steerability for short.

In~\cite[Proposition 2]{Nguyen2016a}, we have shown that the steerability of a qubit state (from Alice's side) is completely determined by the geometry of Alice's steering outcomes $\M_A'$. Specifically, the set of Alice's steering outcomes is called packable if there exists a set of $m$ positive operators $\U_B=\{B_i\}_{i=1}^{m}$ such that $\sum_{i=1}^{m} B_i= (\II^A)'$ and $\M_A' \subseteq \cell(\U_B)$, where $\cell(\U_B) = \{\sum_{i=1}^{m} \beta_i B_i \vert 0 \le \beta_i \le 1 \}$ is called the polyhedral box based on $\U_B$. Then the state $\rho$ is unsteerable from Alice's side with respect to projective measurements if and only if the set of Alice's steering outcomes is packable. To test a necessary condition for steerability, one can simply choose an ansatz $\U_B$ and check if $\M_A' \subseteq \cell(\U_B)$. Note that such an ansatz is always required to give rise to $\cell(\U_B)$ that has the so-called principal vertex $\sum_{i=1}^{m} B_i$ at B's reduced state $(\II^A)'$. In the following, we show how an ansatz can be optimized so that the necessary condition becomes also sufficient.    
\subsection*{The critical radius of local models}
In~\cite{Nguyen2016a}, we have shown that an ansatz $\U_B$ can always be chosen to be a subset of the boundary of the positive cone of $B^H(\H_B)$, which can be represented by a distribution $u(\n)$ on Bob's Bloch sphere $\partial (\S_B)$, normalized by $\int \d \mu (\n)=1$, where $\mu$ is the measure on the Bloch's sphere generated by $u$. The problem of determining if $\M_A'$ is contained in $\cell(\U_B)$ can be made simple by the following lemma.
\begin{lemma}
\label{lem:equator-containing}
For a given state $\rho$ and a given ansatz $\U_B$ then $\M_A' \subseteq \cell (\U_B)$ if and only if $\S_A' \subseteq \cell (\U_B)$.
\end{lemma}
\begin{proof}  
Since $\S_A' \subseteq \M_A'$, the forward direction is obvious. For the reverse direction, assume that $\S_A' \subseteq \cell (\U_B) \cap \P_A'$. As $\M_A$ is the convex hull~\cite{Rockafellar1970a} of $\S_A$ and the two vertices $O$, $\II^A$, and linear maps preserve convex combinations, $\M_A'$ is the convex hull of $\S_A'$ and two vertices $O$ and $(\II^A)'$. Since $\cell(\U_B)$ is a convex set containing $\S_A'$, $O$ and $(\II^A)'$, it also contains $\M_A'$.
\end{proof}

Since $\S_A' \subseteq \P_A'$, the following lemma is obvious.
\begin{lemma}
$\S_A'$ is contained in $\cell (\U_B)$ if and only if $\S_A' \subseteq \cell (\U_B) \cap \P_A'$. 
\end{lemma}

Within  $\P_A'$, $\S_A'$ defines uniquely a metric $d_{\S_A'}$ so that under this metric, $\S_A'$ is precisely the unit ball with center at $\frac{(\II^A)'}{2}$. 
One then deduces that Alice's steering Bloch ball $\S_A'$ is contained in $\cell (\U_B) \cap \P_A'$ if and only if the distance from any point on the boundary $\partial [\cell(\U_B) \cap \P_A']$ to $\frac{(\II^A)'}{2}$ with respect to the metric $d_{\S_A'}$ is at least $1$~\footnote{Another way to see this is making a transformation $\varphi$ on $\P_A'$ to transform $\S_A'$ to a unit ball in the usual metric. Then $\varphi (\S_A')$ is contained in $\varphi [\cell (\U_B)]$ if and only if the usual distance from any point on the boundary of $\partial [\varphi(\cell(\U_B) \cap \P_A')]$ to $\varphi [\frac{(\II^A)'}{2}]$ is at least $1$.}. 
Thus, if we define the \emph{principal radius} of an ansatz $\U_B$ (characterized by a distribution $u$) by
\begin{equation}
r_u (\rho^{A \to B})= \min_{B \in \partial [\cell (\U_B) \cap \P_A']} d_{\S_A'}\left(B,\frac{(\II^A)'}{2} \right),
\label{eq:steering-functional}
\end{equation}
then Alice's steering outcomes $\M_A'$ is contained in $\cell(\U_B)$ if and only if the principal radius is at least $1$.

We define Alice's \emph{critical radius of local models}, or Alice's \emph{critical radius} for short, as
\begin{equation}
R(\rho^{A \to B}) = \sup_u r_u (\rho^{A \to B}).
\label{eq:steering-radius}
\end{equation}
From the above discussion, we obtain
\begin{proposition}\label{prop:steering-radius}
A state $\rho$ is unsteerable from Alice's side if and only if $R(\rho^{A \to B}) \ge 1$.
\end{proposition}
Therefore, the problem of determining steerability has boiled down to the maximization problem~(\ref{eq:steering-radius}) with respect to all possible distributions $u$. Note that the functional under the maximization itself is defined as a conditional minimization~(\ref{eq:steering-functional}). Since a state with degenerate EPR map is separable~\footnote{A degenerate EPR map gives rise to a degenerate steering ellipsoid; according to~\cite{Jevtic2014a,Milne2014a}, the corresponding state is separable.}, and thus unsteerable, we may assume that $\rho^{A \to B}$ is non-degenerate. In this case, $\P_A'$ is a 3D hyperplane and $\S_A'$ is a 3D ellipsoid, the problem~(\ref{eq:steering-functional}) can be considered as a minimization problem over $\partial \left [ \cell(\U_B) \right ]$ with constraint $\P_A'$.

\subsection*{Steerability of non-degenerate T-states}

In this subsection, we illustrate the calculation of Alice's critical radius for the so-called T-states, which are defined~\cite{Horodecki1996a} by
\begin{equation}
\rho= \frac{\II^A}{2} \otimes \frac{\II^B}{2} + \frac{1}{4}\sum_{i,j=1}^{3} T_{ij} \sigma^A_i \otimes \sigma^B_j,
\end{equation}
where $T$ is the $(3 \times 3)$ correlation matrix, which can always be assumed to be diagonal by choosing appropriate bases for $\H_A$ and $\H_B$~\cite{Jevtic2015a}. Again we consider only T-states with non-degenerate correlation matrices, or non-degenerate T-states.

For a T-state, Alice's EPR map
can be presented by $\frac{1}{2} \Theta^T = \frac{1}{2} \oplus \frac{1}{2}T$, where $\frac{1}{2}$ is the contraction along $X_0$ and $\frac{1}{2}T$ acts on the space $X_1X_2X_3$. 
Let $u$ denote a distribution characterizing an ansatz $\U_B$ for the base of a polyhedral box to pack the steering outcomes of T-states. In~\cite{Nguyen2016a}, we have shown that the boundary $\partial [\cell(\U_B)]$ consists of points $\begin{pmatrix} x_0(\lambda, \n_0, g) \\ \b (\lambda,\n_0,g) \end{pmatrix}$ of the form
\begin{equation}
\int \d \mu (\n) [1_{\n_0^{T} \n > \lambda} (\n) + g (\n) 1_{\n_0^{T} \n = \lambda} (\n)] 
\begin{pmatrix} 1 \\ \n \end{pmatrix},
\label{eq:boundary}
\end{equation}
where $\n_0 \in \S_B$, $\lambda \in [-1,1]$ and $g$ is an arbitrary function with values between $0$ and $1$; here $1_X$ denotes the indicator function of a set $X \subseteq \S_B$, $1_X (\n) = 1$ if $\n \in X$ and $1_X(\n) = 0$ everywhere else. For the sake of generality, we do not assume the continuity of $u$; however it is worth noting that if $u$ is continuous, the second contribution in~(\ref{eq:boundary}) vanishes and various statements, e.g., Lemmas~\ref{lem:equator} and~\ref{lem:tangent-principle}, are significantly simplified.

It is also important to note that due to the construction~\cite{Nguyen2016a}, $\n_0$ is a normal vector of a tangent plane at $\b(\lambda,\n_0,g)$ of the cross-section of $\partial [\cell(\U_B)]$ at $x_0$. This remains true even when the surface has many tangent planes at $\b(\lambda,\n_0,g)$ (for example, the surface is a polyhedron and $\b(\lambda,\n_0,g)$ is on one of its edges).

For non-degenerate T-states, $\P_A'$ is the 3D hyperplane $X_0 = \frac{1}{2}$. To constrain the boundary~(\ref{eq:boundary}) to this hyperplane is equivalent to requiring $\lambda$, $\n_0$ and $g$ to satisfy 
\begin{equation}
\int \d \mu (\n) [1_{\n_0^T \n > \lambda} (\n) + g (\n) 1_{\n_0^T \n=\lambda} (\n)] = \frac{1}{2}.
\label{eq:T-state-constraint}
\end{equation} 
Now since $\S_A'$ is central symmetric, for T-states one can also choose $u$ to be central symmetric~\footnote{Since T-states are $Z_2$-invariant, according to~\cite{Wiseman2007a}, the hidden state model can be chosen to be $Z_2$-invariant.}. We first claim
%



\setcounter{lemma}{2}
\begin{lemma} 
\label{lem:equator}
For a central symmetric ansatz $u$, if $\b(\lambda,\n_0,g)$ is a point on the boundary of $\cell(\U_B)$ at the cross section $X_0 = \frac{1}{2}$, then $\b(0,\n_0,g) = \b(\lambda,\n_0,g)$. In other words, we can take $\lambda = 0$ as a solution to the constraint~\eqref{eq:T-state-constraint}.
\end{lemma}

\begin{proof}
The lemma is obvious if $\lambda=0$. Suppose for $\lambda > 0$, one has $\int \d \mu (\n) [1_{\n_0^T \n > \lambda}+ g(\n) 1_{\n_0^T \n = \lambda}]=\frac{1}{2}$. Consider the decomposition of Bob's Bloch sphere to $\n_0^T \n > \lambda$, $\n_0^T \n < - \lambda$ and the boundary $\n_0^T \n = \pm \lambda $, $1 = \int \d \mu (\n) [1_{\n_0^T \n > \lambda}+ g(\n) 1_{\n_0^T \n = \lambda}] + \int \d \mu (\n) [1_{\n_0^T \n < -\lambda} + g (-\n) 1_{\n_0^T \n = -\lambda}] + \int \d \mu (\n) 1_{ -\lambda <\n_0^T \n < \lambda } + \int \d \mu(\n) [1-g(\n)] 1_{\n_0^T \n = \lambda} + \int \d \mu(\n) [1-g(-\n)] 1_{\n_0^T \n = -\lambda}$. Because of the central symmetry of $\mu$, the first two terms are equal, and are $\frac{1}{2}$ by assumption. The last three terms are all non-negative, thus each term should vanish. It also follows that a part of their sum is also zero, $\int \d \mu (\n) [1_{0 < \n_0^T \n < \lambda} + [1-g(\n)] 1_{\n_0^T \n =\lambda} + g(\n) 1_{\n_0^T \n =0}] = 0$. Therefore $\norm{\b (0,\n_0,g)- \b(\lambda,\n_0,g)} \le  \int \d \mu (\n) [1_{0 < \n_0^T \n < \lambda} + [1-g(\n)] 1_{\n_0^T \n =\lambda} + g(\n) 1_{\n_0^T \n =0}] \norm{\n} = 0$, thus $\b (0,\n_0,g)=\b(\lambda,\n_0,g)$. The proof is similar for $\lambda<0$. 
\end{proof}

In the following, we will always assume $\lambda = 0$. The surface $\partial[\cell(\U_B) \cap \P_A']$ in the 3D hyperplane $\P_A'$ is given by 
\begin{equation}
\b (\n_0,g) = \int \d \mu (\n) [1_{\n_0^T \n > 0} (\n) + g (\n) 1_{\n_0^T \n= 0} (\n)] \n,
\label{eq:T-state-boundary}
\end{equation} 
where $\b(\n_0,g)$ denotes the function $\b(0,\n_0,g)$ in~(\ref{eq:boundary}) for simplicity. The constraint~(\ref{eq:T-state-constraint}) is interpreted as a constraint on $g$, $\int \d \mu (\n) [1_{\n_0^T \n > 0} (\n) + g (\n) 1_{\n_0^T \n= 0} (\n)] = \frac{1}{2}$.

The following lemma allows us to compute the principal radius $r_u (\rho^{A \to B})$ using the parametrization~(\ref{eq:T-state-boundary}).
\begin{lemma}
\label{lem:tangent-principle}
The principal radius of an ansatz $u$ with respect to the EPR map $\rho^{A \to B}$ of a T-state can be found by
\begin{equation}
r_u (\rho^{A \to B}) = \min_{\n_0}  \frac{ 2 \n_0^{T} \b (\n_0,g)}{ \sqrt{\n_0^T T^2 \n_0}},
\label{eq:linear-functional}
\end{equation}
where $\b(\n_0,g)$ is defined by equation~(\ref{eq:T-state-boundary}). Note that the expression on the right side of~(\ref{eq:linear-functional}) does not depend on $g$. 
\end{lemma}

\begin{figure}[h!]
\includegraphics[width=0.4\textwidth]{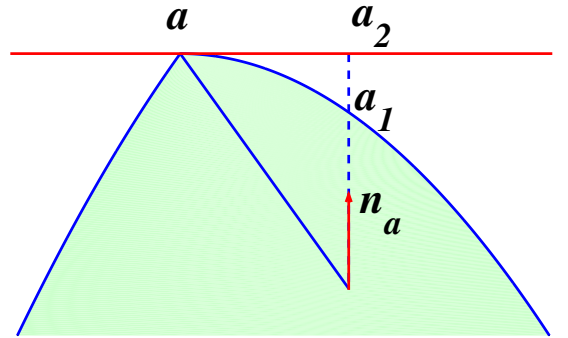}
\caption{Two-dimensional illustrations of $\a$, $\n_a$, $\a_1$ and $\a_2$. Note that the point $\a$ was chosen to have multiple tangent planes.}
\label{fig:minr}
\end{figure}

\begin{proof}
Note that the $2T^{-1}$ transformation brings $\S_A'$ to $\S_A$, and thus brings the distance $d_{\S_A'}$ to the usual one. Let us apply the $2T^{-1}$ transformation to $\cell(\U_B)$ and let $\a \in 2T^{-1} \partial \left [ \cell(\U_B) \cap \P_A' \right ]$, then $r_u (\rho^{A \to B})= \min_{\a} \norm{\a}$. Let $\n_{\a}$ be the normal vector of one of the planes that are tangent to $ 2T^{-1} \partial \left [ \cell(\U_B) \cap \P_A' \right ]$ at $\a$, we then show that
\begin{equation}
r_u (\rho^{A \to B}) = \min_{\a} (\n_{\a}^{T} \a). 
\label{eq:normal-principle}
\end{equation}
Indeed, let $\a_1$ and $\a_2$ be the intersection points of the ray $\n_{\a}$ with $2 T^{-1} \partial [\cell(\U_B) \cap \P_A']$ and the tangent plane at $\a$ respectively; see Figure~\ref{fig:minr}. Since $2T^{-1} \partial [\cell(\U_B) \cap \P_A']$ is the boundary of a convex domain, $\norm{\a_1} \le \norm{\a_2} = \n_{\a}^{T} \a \le \norm{\a}$. The first inequality implies $ r_u (\rho^{A \to B}) \le  \min_{\a} (\n_{\a}^{T} \a)$, while the second inequality implies $\min_{\a} (\n_{\a}^{T} \a) \le r_u (\rho^{A \to B})$; thus $r_u (\rho^{A \to B})= \min_{\a} (\n_{\a}^{T} \a)$. 

As we noted, the point $\b(\n_0,g)$ on $\partial [\cell(\U_B) \cap \P_A']$ admits $\n_0$ as the normal vector of one of the tangent planes at $\b(\n_0,g)$. As normal vectors transform covariantly under the $2T^{-1}$ transformation, we deduce that the normal vector of one of the tangent planes of $2T^{-1} \partial \left [ \cell(\U_B) \cap \P_A' \right ]$ at $\a= 2 T^{-1} \b(\n_0,g)$ is $\n_a = T \n_0/\sqrt{\n_0^T T^2 \n_0}$. Substituting $\a$ and $\n_{\a}$ into~\eqref{eq:normal-principle}, the dependence on $g$ is dropped out and one obtains~(\ref{eq:linear-functional}).
\end{proof}


Now Jevtic \emph{et al.}~\cite{Jevtic2015a} have shown that if one chooses $u=J$ with
\begin{equation}
J (\n)= \frac{N_T}{(\n^{T} T^{-2} \n)^2},
\label{eq:Jevtic}
\end{equation}
where $N_T$ is the normalization factor, which can be calculated explicitly~\cite{Jevtic2015a}, the surface $\b(\n_0)$ is analytically tractable,
\begin{equation}
\b (\n_0,g)= \pi N_T \abs{  \det (T)}  \frac{T^2 \n_0}{\sqrt{\n_0^T T^2 \n_0}}, 
\end{equation}
which is independent of $g$ since the ansatz is continuous.
This gives
\begin{equation}
\frac{\n_0^{T} \b (\n_0,g)}{\sqrt{\n_0^T T^2 \n_0}} = \pi N_T \abs{ \det (T)},
\end{equation}
and thus~(\ref{eq:normal-principle}) implies
\begin{equation}
r_J (\rho^{A \to B})= 2 \pi N_T \abs{\det (T)}.
\end{equation}
We are going to show that the ansatz (\ref{eq:Jevtic}) is in fact optimal in the sense that it maximizes $r_u (\rho^{A \to B})$ with respect to $u$. To this end, we show that $r_{J + v} (\rho^{A \to B}) \le r_J (\rho^{A \to B})= 2 \pi N_T \abs {\det (T)}$ for all $v$ central symmetric and normalized by $\int \d S (\n) v (\n)=0$, where $\d S (\n)$ is the area measure on Bob's Bloch sphere. Indeed, substituting $u=J+v$ into~\eqref{eq:linear-functional} (and choosing $g=1$ to simplify the notation), we find
\begin{align}
r_{J+v} (\rho^{A \to B}) &= 2 \pi N_T\abs{ \det (T) }+ \min_{\n_0} I(\n_0),
\label{eq:Jevtic-principal-radius}
\end{align}
with 
\begin{align}
I(\n_0) = \int \d S( \n) v (\n) 1_{\n_0^T \n \ge 0} \frac{2\n_0^T \n}{\sqrt{\n_0^T T^2 \n_0}}.
\end{align}
If one chooses $f(\n_0)= \frac{1}{2}\sqrt{\n_0 T^2 \n_0}$ then one finds that $\int \d S (\n_0) f(\n_0) I(\n_0)=0$ because it is equal to $\int \d S(\n) v (\n) \int \d S (\n_0) 1_{\n_0^T \n \ge 0} \n_0^T \n$, which vanishes because $\int \d S (\n_0) 1_{\n_0^T \n \ge 0} \n_0^T \n$ does not depend on $\n$ and
$ \int \d S(\n) v (\n)= 0$. Since $f(\n_0) > 0$, it is then clear that $\min_{\n_0} I(\n_0) \le 0$. Together with~(\ref{eq:Jevtic-principal-radius}), this implies $r_{J + v} (\rho^{A \to B}) \le 2\pi N_T \abs{\det (T)}$. By definition, the critical radius for a T-state is $R(\rho^{A \to B})= 2 \pi N_T \abs{\det (T)}$. Therefore a T-state is unsteerable if and only if $2 \pi N_T \abs{\det (T)} \ge 1$ as conjectured by Jevtic~\textit{et al.}~\cite{Jevtic2014a}. 

Since Werner states~\cite{Werner1989a} are special T-states, using this result one can easily recover the known necessary and sufficient condition for Werner states to be steerable~\cite{Wiseman2007a}. Moreover, it is also worth noting that if the center of Bob's steering ellipsoid (which is refered to as \emph{Alice's} steering ellipsoid in~\cite{Jevtic2014a}) is at the center of Alice's Bloch sphere, then the state can be brought into a T-state by a transformation without changing its steerability from Alice's side~\cite[Lemma 1]{Bowles2016a}; thus its steerability from her side can also be fully characterized.   

\section{Conclusion}
We have defined the critical radius of local models and shown that a state over a two-qubit system is steerable from Alice's side if and only if her critical radius of local models is less than $1$. We explicitly calculate the critical radius of local models for a T-state, thereby proving the conjecture on the sufficient and necessary condition for steerability of T-states. In our opinion, an important question for future development is to formulate a calculation of the critical radius for an arbitrary two-qubit state, either analytically or numerically. 

%
\begin{acknowledgements}
Discussions with Michael Hall, Sania Jevtic, Xuan Thanh Le, Antony Milne have been very helpful. In particular, Sania pointed out to us the separability of states with degenerate EPR maps. Antony and Michael drew our attention to the result on the steerablity of states that have the center of Bob's steering ellipsoid at the center of Alice's Bloch sphere. Comments from Sania and Antony have also significantly improved the readability of our manuscript. 
\end{acknowledgements}
\bibliography{../bibtex/quantum-steering}
\end{document}